\documentclass{article}

\usepackage{latexsym}
\usepackage{amsthm}
\usepackage{amsmath}
\usepackage{amssymb}
\usepackage{xcolor}

\newcommand{\mon}{\textsc{MOn-1qfa}}
\newcommand{\lat}{\textsc{LQFA}}
\newcommand{\lats}{\textsc{LQFA}s}
\newcommand{\mons}{\textsc{MOn-1qfa}s}
\newcommand{\lmo}{\textbf{LMO}}
\newcommand{\pmo}{\textbf{PMO}}
\newcommand{\pt}{\textbf{PT}}
\newcommand{\lipt}{\text{liId}\textbf{PT}}
\newcommand{\id}{\text{liId}}

\newcommand{\C}{\mathbb{C}}

\newcommand{\R}{\mathbf{R}}
\newcommand{\J}{\mathbf{J}}
\newcommand{\bg}{\mathbf{BG}}

\newtheorem{definition}{Definition}
\newtheorem{theorem}{Theorem}
\newtheorem{proposition}{Proposition}
\newtheorem{corollary}{Corollary}
\newtheorem{lemma}{Lemma}

\begin{document}

\title{Algebraic Characterization of the Class of Languages recognized by Measure Only Quantum Automata}
\date{September 2013}

\author{Carlo Comin
\thanks{The author wish to thank Prof. Alberto Bertoni for having introduced him to the problem and 
for the stimulating discussions and research directions that finally lead to the results of this paper,  
Dr. Maria-Paola Bianchi for the kind collaboration on the structural drawing-up of the article. 
Also wish to thank the anonymous referees for their careful review and for their valuable comments, 
which improved both the content and quality of the manuscript.}\\
Dipartimento di Informatica, \\
Universit\'a degli Studi di Milano\\ via Comelico 39/41, Milano, Italy.\\
carlo.comin{@}studenti.unimi.it \\
}
\maketitle


\begin{abstract}
 We study a model of one-way quantum automaton where only measurement operations are allowed (\mon). 
We give an algebraic characterization of $\lmo(\Sigma)$, showing that the syntactic monoids 
of the languages in $\lmo(\Sigma)$ are exactly the $J$-trivial literally idempotent syntactic monoids, 
where $J$ is the \emph{Green}'s relation determined by two-sided ideals.
We also prove that $\lmo(\Sigma)$ coincides with the literal variety of literally idempotent piecewise testable languages. 
This allows us to prove the existence of a polynomial-time algorithm for deciding whether a regular language 
belongs to $\lmo(\Sigma)$ and to discuss definability issues in terms of the existential first-order logic $\Sigma_1[<]$ 
and the linear temporal logic without the next operator LTLWN. 
\end{abstract}


\section{Introduction}
This paper gives a characterization of the class of languages recognized by a model of quantum automata, 
by using tools from algebraic theory, in particular, varieties of languages and syntactic monoids.
Many models of one-way quantum finite automata are present in the literature: the oldest is 
the Measure-Once model \cite{BC01,BP02}, characterized by unitary evolution operators and 
a single measurement performed at the end of the computation. On the contrary, in other models, 
evolutions and measurements alternate along the computation \cite{Aal06, KW97}. 
The model we study is the Measure-Only Quantum Automaton (\mon), introduced in \cite{BMP10}, 
in which we allow only measurement operations, not evolution.
All these quantum models are generalized by Quantum Automata with Control Language \cite{BMP03}.
We also remind that the use of quantum measurements as computational steps is an active area of research, 
falling under the umbrella-name of MBQC (measurements-based quantum computation); 
we cite for instance the teleportation based model and the one-way quantum computer of cluster state 
computation \cite{GC99, Le04, Ni03}.\\\\
Originally, the \mon\, model has been introduced by \emph{Bertoni, et al.}, upon the framework of trace theory and 
free partial commutative monoids \cite{BMP10, BBMS81, DR95}, giving evidence that \emph{Mauri}'s 
approach to link trace theory to that of free partially commutative monoids can be useful also in the area of 
\emph{quantum computing}. Those insights finally lead to the results contained in \cite{BMP10},
 in which the authors proved that every formal series induced by the probabilistic behaviour of a \mon\, 
 admits a linear representation with projectors, moreover, that the whole family of such formal series is closed 
 under $f$-complement and Hadamard product and, as a main result, that the class of languages recognized 
 by \mons\, over compatibility alphabets $(\Sigma, E)$ is a boolean algebra of recognizable languages with finite variation. 
\\\\ In the present work we study \mons\, over finite non-empty set alphabets $\Sigma$, thus relaxing the requirement which forces us to fix a 
compatibility relation $E$. We shall denote the class of formal languages recognized by \mons\, over $\Sigma$ as $\lmo(\Sigma)$. 
By doing so we study a particular class of the languages studied in \cite{BMP10} and 
prove a characterization result in terms of the properties satisfied by syntactic monoids of those languages; 
in particular we lean upon the notion of $J$-triviality (where $J$ is one of the \emph{Green}'s relation, 
determined by two-sided monoid ideals) and the literal idempotency relation, 
i.e. $\sigma^2\sim \sigma$ for every letter $\sigma\in\Sigma$. 
This  general approach, to study properties of formal languages as a reflex-effect of the algebraic properties and equations
 satisfied by their syntactic monoids, is a well-established and celebrated principle in the algebraic theory of automata and 
 formal languages, having roots in finite semigroup and finite monoid theory. 
 The pillars of the formerly mentioned algebraic framework stand in the seminal works of \emph{Eilenberg}, 
 \emph{Sh\"utzenberger}, \emph{Simon}, \emph{Straubing}, \emph{Pin}, et al.; 
 for instance, we would like to cite it here the celebrated 
Eilenberg Variety Theorem \cite{Ei76}, the algebraic characterization of the piecewise testable languages in 
 terms of $J$-trivial syntactic monoids  given by \emph{Simon} \cite{Si75}, the algebraic characterization of the star-free languages
  by means of aperiodic syntactic monoids of \emph{Sh\"utzenberger} and the broad algebraic theory of varieties of languages and 
 pseudovarieties of monoids developed by \emph{Straubing} \cite{HS94}, \emph{Pin} \cite{Pi86}, \emph{et al.} 
 Those powerful algebraic masterworks have already 
 been used successfully in order to characterize classes of languages recognized by finite-state quantum devices; 
 as in the work of \emph{Ambainis, et al.} \cite{Aal06}, in which the authors established intimate ties between the family  
 of so-called Latvian automata and the pseudovariety $\bg$ of block-group syntactic monoids. 
As an aftermath of our mathematical journey, with this work we stand to defend that a formal language $L$ is
 recognized by some \mon\, if and only if it is a literal idempotent piecewise testable language, that is, 
 if and only if its syntactic monoid is $J$-trivial and it is such that its associated syntactic homomorphism literally satisfies 
 the idempotency pseudoidentity $x^2 = x$; this particular literal variety of the literally idempotent regular languages has already
 been studied algebraically by Kl\'ima and Pol\'ak in \cite{KP08}. As a corollary, we solve the polynomial-time 
 decidability question for $\lmo(\Sigma)$\, by proving the existence of a polynomial (quadratic) time algorithm for 
 deciding whether a regular language belongs to $\lmo(\Sigma)$. 
 We also discuss the definability of languages 
 recognized by \mons\, in terms of \emph{easy} formulas of the existential first-order logic $\Sigma_1[<]$ and 
 \emph{easy} formulas of the linear temporal logic without the next operator LTLWN. We would like to remind that the 
 logical approach to automata and formal languages is also a well-established area of mathematical research, 
 pioneered by the work of \emph{B\"uchi} \cite{Bu60}; to the best of our knowledge, he was the first to conceive and prove that 
 the regular languages are exactly those languages definable in the monadic second-order logic. 
 We highly remark that it is an extremely beauty asset of automata and formal language theory to be a 
 bridge between quantum and algebraic insights. 
\section{Preliminaries}

\subsection{Notation}
We shall denote by $\mathbb{N} := \{0, 1, 2, \cdots \}$ the set of natural numbers and by $\mathbb{N}_0 := \{1,2,\cdots\}$ 
the set of the positive ones. A semigroup is any set $S$ together with a binary operation $\cdot$ that is associative. 
A monoid $M$ is a semigroup containing an identity element $1_M$. If $S$ is a semigroup with no identity, we denote by $S^1$ the monoid obtained by $S\cup\{1_M\}$ such that $1_M$ is an identity in $S$. If $S$ is a monoid, then $S^1=S$.
If $S,S'$ ($M,M'$) are semigroups (monoids), and the map $\varphi:S\rightarrow S'$ ($\varphi:M\rightarrow M'$) is such 
that $\varphi(x\cdot x')=\varphi(x)\cdot \varphi(x')$ for every $x\in S$ ($x\in M$), then we say that $\varphi $ is a semigroup (monoid) homomorphism; we denote by $\varphi: X\hookrightarrow Y$ injective homomorphisms, by $\varphi: X \twoheadrightarrow Y$ 
surjective and by $\varphi: X \leftrightarrow Y$ bijective homomorphisms.\\
Given a finite alphabet $\Sigma$, we write $\Sigma^*$ to denote the free monoid generated by $\Sigma$. The free monoid $
\Sigma^*$ includes all possible words whose letters belong to $\Sigma$ and the empty word $\epsilon$. If $w\in\Sigma^*$ is a word, we denote its length by $|w|$ and its letters by $w_1, \ldots, w_{|w|}$. A subset $L\subseteq\Sigma^*$ is a formal language. A deterministic finite-state automaton is a tuple $A : = \langle \Sigma, Q, \delta, q_0, F \rangle$ where 
the transition function $\delta$ is from $\Sigma\times Q$ to $Q$ and $L_A:=\{w\in\Sigma^* | \delta(w)\in F\}$, where $\delta(w):=\delta(\cdots\delta(\delta(q_0, w_1), w_2), w_n)$ for any $w:=w_1\cdots w_n$ s.t. $|w|=n$, is the language recognized by $A$. A language recognized by some finite-state automaton is regular.\\ 
Let $L$ be a regular language and let $A_L : = \langle \Sigma, Q, \delta, q_0,
F \rangle$ be the minimal deterministic automaton recognizing $L$. For any word $w=w_1\cdots w_n\in\Sigma^*$, 
we define its variation as the cardinality 
$\text{var}_L(w) :=  \#\{0\leq k < n \mid \delta(w_1\cdots w_k)\neq \delta(w_1\cdots w_{k+1})\}$.
We say that $L$ has \emph{finite variation} if and only if $ \text{sup}_{x\in\Sigma^*} \text{var}_L(x) < \infty $.
\subsection{Linear algebra for quantum systems}
In this section we briefly outline some notions of linear algebra; this allows us to describe the concepts of quantum observable, 
quantum measurement and, more generally, quantum finite-state computing device. This summary follows the one given in 
\cite{BMP10}.
We denote the field of complex numbers by $\mathbb{C}$. Given a complex number $z\in\mathbb{C}$, its complex conjugate is 
denoted by $z^*$, and its modulus by $|z|=\sqrt{zz^*}$. We denote by $\mathbb{C}^{m \times n}$ the set of $m\times n$ matrices 
with complex valued entries. For any $M\in\mathbb{C}^{m\times n}$ and for any $1\leq i\leq m$ and $1\leq j\leq n$, we denote 
by $M_{ij}$ or $(M)_{ij}$ the $(i,j)$-th entry of $M$. The adjoint matrix of $M\in\mathbb{C}^{m\times n}$ is denoted by 
$M^{\dag}$ and it is defined by $M^{\dag} = (M^*)^T = (M^T)^*$, provided that $M^*$ is defined by $M^*_{ij} := (M_{ij})^*$ and that 
$M^T$ denotes matrix transposition. 
The trace of a square matrix $M\in\mathbb{C}^{n\times n}$ is given by $\text{Tr}(M):=\sum_{i=1}^{n} M_{ii}$. 
If $A\in\mathbb{C}^{m\times n}$ and $B\in\mathbb{C}^{p\times q}$, then their direct sum 
is the $(m+p)\times (n+q)$ matrix defined as follows: 
\[
A \oplus B := \left(
\begin{array}{c c}
A & \mathbf{0} \\
\mathbf{0} & B \\
\end{array}
\right) \;\;\; 
\] where $\mathbf{0}$ denotes null matrices of suitable dimensions. 
An Hilbert space of finite dimension $m$ is the linear space $\mathbb{C}^{1\times m}$ equipped with sum and product by 
elements in $\mathbb{C}$, in which the inner product is defined as $(\pi, \xi):=\pi\xi^{\dag}$. If $(\pi, \xi)=0$ then we say 
that $\pi$ is orthogonal to $\xi$. The norm of vector $\pi$ is given by $||\pi||:=\sqrt{(\pi, \pi)}$. Two subspaces $X, Y$ are 
orthogonal if any vector in $X$ is orthogonal to any vector in $Y$ and, in this case, the linear space generated by $X\cup Y$ is 
denoted by $X\oplus Y$. A matrix $M\in\mathbb{C}^{m\times m}$ can be view as a morphism $\pi\mapsto \pi M$ of the 
Hilbert space $\mathbb{C}^{1\times m}$ in itself and it is said to be Hermitian whenever $M=M^{\dag}$. 
Given an Hermitian matrix $O$, we denote 
by $\lambda_1, \ldots, \lambda_s$ its eigenvalues and by $E_1, \ldots, E_s$ the corresponding eigenspaces. It is well-known 
that each eigenvalue $\lambda_k$ is real, that $E_i$ is orthogonal to $E_j$ for any $i\neq j$ and that 
$E_1\oplus \cdots \oplus E_s = \mathbb{C}^{1\times m}$. Each vector $\pi$ can be uniquely decomposed such as 
$\pi = \pi_1 + \cdots + \pi_s $, where $\pi_j \in E_j$; the linear transformation $\pi\mapsto\pi_j$ is called projector $P_j$ on the 
subspace $E_j$. A linear operator is a projector if and only if it is Hermitian and idempotent. Every Hermitian matrix $O$ is 
uniquely determined by its eigenvalues and its eigenspaces, or by its projectors. By the spectral decomposition theorem it holds 
that, for some $s\in\mathbb{N}_0$, $O=\sum^{s}_{i=1} \lambda_{i} P_i$ and we denote by 
$V(O):=\{\lambda_i\}_{i=1}^s$ the spectrum of $O$. 
Given the set $e_1, \ldots, e_m\in\{0,1\}^m$ of pure states, a quantum state is a superposition 
$\pi = \sum_{k=1}^{m} \pi_k e_k$ where the coefficients $\pi_k$ are complex amplitudes and $||\pi|| = 1$. A 
quantum observable is represented by an Hermitian operator $O=\sum_{i=1}^s \lambda_i P_i$ where the spectrum $V(O)$ is the set of possible results of a measurement of $O$. A measurement of $O$ on $\pi$ will return $\lambda_j$ with probability 
$|| \pi P_j ||^2$ and the state after the quantum measurement becomes $\pi P_j/||\pi P_j||^2$.
\subsection{Varieties of formal languages, pseudovarieties of finite monoids and literal idempotency}
This section is devoted to the recall of some general definitions and results from the algebraic 
theory of automata and formal languages. For more details, we refer the reader to, e.g. \cite{Ei76,Pi86}.
\begin{definition}[Syntactic monoid]
Let $L\subseteq\Sigma^*$ be a language over the alphabet $\Sigma$. The syntactic congruence $\sim_L$ w.r.t. 
$L$ is defined as follows: for every $x,y\in\Sigma^*$, $x\sim_L y$ if and only if for every $a,b\in\Sigma^*$ it holds that $axb\in L \iff ayb\in L
$. For any language $L\subseteq\Sigma^*$, we say that the quotient monoid $M(L):=\Sigma^*/ \sim_L$ is the syntactic monoid of $L$. 
\end{definition}
\begin{definition}[M-pseudovariety]
We say that a class of finite monoids $\mathbf{M}$ is a pseudovariety if and only if the following three conditions holds:
(i) If $M\in \mathbf{M}$ and $N$ is a submonoid of $M$, then $N\in\mathbf{M}$. 
(ii) $M\in\mathbf{M}$ and $Q$ is a homomorphic image, i.e. quotient monoid, of $M$, then $Q\in\mathbf{M}$. 
(iii) If $M,N\in\mathbf{M}$, then the direct product monoid $M\times N$ is also in $\mathbf{M}$.
\end{definition}
\begin{definition}[Literal pseudovariety \cite{KP08}] Let $\mathbf{H}$ be a class of surjective homomorphisms from free monoids over non-empty sets onto finite 
monoids. Then $\mathbf{H}$ is a literal pseudovariety if it is closed with respect to the homomorphic images, literal 
substructures and products of finite families; formally, the following three conditions must be satisfied: 
(i) For each $(\varphi : \Sigma^*\twoheadrightarrow M)\in\mathbf{H}$ and surjective monoid homomorphism 
$\sigma:M\twoheadrightarrow N$, it holds $\sigma\varphi\in \mathbf{H}$.
(ii) For each $(\varphi : \Sigma^*\twoheadrightarrow M)\in\mathbf{H}$ and for each free monoid homomorphism 
$f:\Gamma^*\rightarrow \Sigma^*$ such that $f(\Gamma)\subseteq \Sigma$, it holds 
$(\varphi f : \Gamma^* \twoheadrightarrow (\varphi f)(\Gamma^*))\in\mathbf{H}$.
(iii) For each non-empty set $\Sigma$, the mapping of $\Sigma^*$ onto the one element monoid $\{1\}$ is in $\mathbf{H}$, 
moreover, for each $\varphi:\Sigma^*\twoheadrightarrow M, \psi:\Sigma^*\twoheadrightarrow N$, the natural homomorphism 
of $\Sigma^*$ onto $\Sigma^* / (\text{ker}\,\varphi \cap \text{ker}\,\psi$) is in $\mathbf{H}$.
\end{definition}

\begin{definition}[*-variety and literal variety of languages] Let $\mathbf{M}$ be a class of monoids and let $\Sigma$ be an alphabet. We denote by $V_{\Sigma}(\mathbf{M})$ 
the class of regular languages on $\Sigma$ having syntactic monoid in $\mathbf{M}$. 
We say that a class of regular languages $V:\Sigma\rightarrow 2^{\Sigma^*}$ is a $*$-variety of \emph{Eilenberg} if $V$
 is closed under boolean operations, right and left quotient, and inverse homomorphism. Replacing closure under inverse 
 homomorphism by closure under inverse literal homomorphism, we get the notion of literal variety of languages. More precisely, 
 closure under literal homomorphisms holds whenever the following condition is satisfied: for each alphabets $\Gamma$ and 
 $\Sigma$ and a free-monoid homomorphism $f:\Gamma^*\rightarrow \Sigma^*$ such that $f(\Gamma)\subseteq \Sigma$, then 
 $L\in V(\Sigma)$ implies $f^{-1}(L)\in V(\Gamma)$.
 \end{definition}
%
%
%
A fundamental result is due to \emph{Eilenberg}, who showed that there exists a bijection $V$ from the psuedovarieties of 
monoids and the $*$-varieties of formal languages \cite{Pi86}. 
This result, known as the Eilenberg Variety Theorem, has been extended in literature  
in many ways; for instance, \emph{Straubing} considered the more general notion of $\mathbb{C}$-variety 
\cite{St02} and the work of \emph{\'Esik, et.al.} \cite{EI03}, \cite{EL03} focused on literal varieties and the corresponding links with literal pseudovarieties. 
Due to \emph{Kunc} \cite{Ku03} we also have equational logic for those classes of languages.
This result of \emph{Eilenberg} is thus prominent and of remarkable importance in the algebraic theory of automata and formal languages.
\newline We also recall some fundamental equivalence relations studied by \emph{Green} in 1951 \cite{Pi86}. They allow us to describe some 
$*$-varieties of languages arising from the notion of triviality with respect to \emph{Green}'s equivalence classes. 
For instance, we introduce the pseudovariety of $J$-trivial monoids and the $*$-variety of the piecewise testable languages, as defined below.
\begin{definition}[Green's relations, \cite{Pi86}]
We denote by $L, R$ and $J$ the \emph{Green}'s relations determined by left, right and two-sided ideals, respectively.
In more detail, let $S$ be a semigroup, for any $a,b\in S^{1}$ we have 
\begin{enumerate}
\renewcommand\labelenumi{\theenumi}
\renewcommand{\theenumi}{(\roman{enumi})}
\item $a L b \iff S^1 a = S^1 b$.
\item $a R b \iff aS^1 = b S^1$.
\item $a J b \iff S^1 a S^1 = S^1 b S^1$.
\end{enumerate}
Let $S$ be a semigroup and $G$ be one of \emph{Green}'s relations, we say that $S$ is $G$-trivial if and only if $aGb$ implies $a=b$, for every $a,b\in S$.
In this paper we denote by $\R $ the pseudovariety of $R$-trivial finite monoids and by $\J$ the pseudovariety of $J$-trivial finite monoids.
We also define $\overline\J$ as the class of $J$-trivial finite monoids $M$ such that every surjective 
homomorphism $\varphi:\Sigma^*\twoheadrightarrow M$, from the free monoid generated by a non-empty set onto M, literally satisfies the idempotency pseudoidentity $x^2=x$, i.e. $\varphi(\sigma)\varphi(\sigma)=\varphi(\sigma)$, for every $\sigma\in\Sigma$.\end{definition}
\begin{definition}[Literally idempotent piecewise testable languages]
We say that a language $L\in\Sigma^*$ is literally idempotent if and only if for all $x,y\in\Sigma^*$ and $a\in\Sigma$, 
$xa^2y\in L\Leftrightarrow xay\in L$. We say that $L$ is a piecewise testable language if and only if it lies in the boolean 
closure of the following class of languages, defined for each $k\geq 0$:
\[\Sigma^* a_1 \Sigma^* a_2 \Sigma^*\cdots \Sigma^* a_k \Sigma^*, \text{ for } a_1, a_2, \ldots, a_k\in\Sigma \]
Moreover, we say that $L$ is literally idempotent piecewise testable if and only if it lies
in the boolean closure of the following class of languages, defined for each $k\geq 0$: 
\[\Sigma^* a_1 \Sigma^* a_2 \Sigma^*\cdots \Sigma^* a_k \Sigma^*, \text{ for } a_1, a_2, \ldots, a_k\in\Sigma \text{ and } 
 a_i \neq a_{i+1} \text{ for every } 1 \leq i < k \]
We denote by $\pt$ the class of the piecewise testable languages, by $\id$ the class of literally idempotent languages and by $\lipt$ the class of the literally idempotent piecewise testable languages.
\end{definition}
One of the first known instances of the Eilenberg Variety Theorem is based on the $J$ relation.
In fact, it has been proved by \emph{Simon} the following remarkable characterization result: 
\begin{theorem}[Simon, \cite{Si75}]
$L$ is a piecewise testable language 
if and only if its syntactic monoid is $J$-trivial. 
\end{theorem}
In \cite{KP08}, \emph{Kl\'ima} and \emph{Pol\'ak} proved the following characterization for the literally idempotent piecewise testable 
languages.
\begin{theorem}[Kl\'ima, Pol\'ak, \cite{KP08}]\label{thm:klimapolak} Let $L\subseteq\Sigma^*$ be a formal language. 
Then the following propositions are equivalent: 
\begin{enumerate}
\renewcommand\labelenumi{\theenumi}
\renewcommand{\theenumi}{(\roman{enumi})}
\item $L$ lies in the boolean closure of the following class of languages: \\
$\Sigma^* a_1 \Sigma^* a_2 \Sigma^*\cdots \Sigma^* a_k \Sigma^*$ for $a_1, a_2, \ldots, a_k\in\Sigma$ and 
 $a_i \neq a_{i+1}$ for every $1 \leq i < k $,\\ that is to say 
 $L\in\lipt$.
\item $L$ is picewise testable and literally idempotent, that is to say $L\in \pt\cap\id(\Sigma) = V_{\Sigma}(\mathbf{J})\cap \id(\Sigma)$.
\item the syntactic monoid of $L$ is $J$-trivial and satisifies the pseudoidentity $\sigma^2=\sigma$ literally, \\ we denote this fact by $L\in V_{\Sigma}\left(\overline{\J}\right)$.
\end{enumerate}
\end{theorem}
\subsection{\mons\, and \lats}
The \mon\, model has been originally introduced in the context of free partial commutative monoids with idempotent 
generators \cite{BMP10}. 
Let $(\Sigma, E)$ be a compatibility alphabet, that is a finite simple graph over a finite non-empty set $\Sigma$. Let us denote by $\underline{E}$ the least congruence containing
$(cc,c)$ for all $c\in\Sigma$ and $(ab,ba)$ for all $a,b\in\Sigma$ such that $(a,b)\in E$. We say that $\Sigma^*/\underline{E}$ is a 
free partially commutative monoid with idempotent generators and we denote it by $\text{FI}(\Sigma, E)$. 
Each element $t\in\text{FI}(\Sigma, E)$, that is an equivalence class of $\underline{E}$, shall be seen as a language $t\subseteq\Sigma^*$.\\
Following \cite{BMP10}, a \mon\ over the compatibility alphabet $(\Sigma, E)$ is a tuple of the form 
$A := \langle (\Sigma, E) \cup \{\#\},$ $(O_c)_{c\in \Sigma\cup\{\#\}}, \pi_0, F \rangle$ 
where the complex $m$-dimensional vector $\pi_0\in\C^{1\times m}$, with unitary norm $||\pi_0||=1$, 
is called the quantum initial state of $A$ and for every $c\in\Sigma$, $O_c\in\C^{m\times m}$
is (the representative matrix of) an Hermitian operator and denotes an observable. Also, 
the subset $F\subseteq V(O_{\#})$ of the eigenvalues of $O_{\#}$ is called the spectrum of the quantum final accepting 
states of $A$.\\ Here we study \mons\, over finite non-empty set alphabets $\Sigma$, thus relaxing the requirement which forces us to fix a 
compatibility relation $E$.
\begin{definition}[\mons\, over $\Sigma$]\label{def:monOverSigma}
Let $\Sigma$ be a finite non-empty set. A \mon\ over the alphabet $\Sigma$ is a tuple of the form 
$
A := \langle \Sigma \cup \{\#\},$ $(O_c)_{c\in
  \Sigma\cup\{\#\}}, \pi_0, F \rangle
$. 
The complex $m$-dimensional vector $\pi_0\in\C^{1\times m}$, with unitary norm $||\pi_0||=1$, 
is called the quantum initial state of $A$.
For every $c\in\Sigma\cup\{\#\}$, $O_c\in\C^{m\times m}$
is (the representative matrix of) an Hermitian operator and denotes an observable.
The subset $F\subseteq V(O_{\#})$ of the eigenvalues of $O_{\#}$ is called the spectrum of the quantum final accepting 
states of $A$. We also say that $A$ is of finite dimension $m\in\mathbb{N}_0$. 
\end{definition}
\begin{definition}[\mons\, language recognition]
Let $A := \langle \Sigma \cup \{\#\},$ $(O_c)_{c\in\Sigma\cup\{\#\}}, \pi_0, F \rangle$ be a \mon\, over $\Sigma$.
Starting from the initial state $\pi_0$, at each step of the computation an input letter $c$ is read, 
then the system is measured by applying observable $O_c$. At the end of the input word $w$, we 
measure the system with the end-marker observable $O_{\#}$ and we probabilistically accept or reject $w$ according to whether this last measure is in $F$ or not. 
We formalize this process by means of density matrices.
For every $c\in\Sigma\cup\{\#\}$ let $P_j(c)$ to be the $j$-th projector of $O_c$, and let $k(c)$ to be the cardinality of all such projectors, that is to say the cardinality of $V(O_c)$. 
Then, for every word $w\in\Sigma^*$, we define the following density matrix $\sigma(w)$:
\[
 \sigma(w) = \left\{
  \begin{array}{l l}
    \pi_0^{\dag}\pi_0 &\text{ se $w=\epsilon$}\\
    \sum_{j=1}^{k(c)} P_j(c)\sigma(y) P_j(c)& \text{ if $w=cy$ for some $c\in\Sigma, y\in\Sigma^*$} \\
  \end{array} \right.
\]
The probability that $A$ accepts $w$ is defined as $p_A(w) := \text{Tr}\left( \sum_{r_j\in F} P_{\text{j}}(\#)\sigma(w) P_{\text{j}}(\#)\right)$ 
where $r_j\in F$ vary among final accepting eigenvalues.
If $\lambda$ is a positive real number, we say that the language $L_A:=\{w\in\Sigma^* | p_A(w) > \lambda\}$ is recognized by $A$ with $\lambda$ cut-point. 
The cut-point $\lambda$ is said to be isolated whenever there exists a positive real $\delta$ such that 
$|p_A(w)-\lambda| \geq \delta$ for every $w\in\Sigma^*$. 
\end{definition}
As we shall see, the \mon\, model is intimately related to another model of one-way quantum finite-state automata, known as the Latvian model, introduced by $\emph{Ambainis, et al.}$ \cite{Aal06}. This fact will turn out to be a cornerstone in order to prove our result. In this section we review the Latvian model and its known characterization.
\begin{definition}[\lats]
A Latvian Automaton (\lat) with $m$ elementary pure states is a 5-tuple of the following form
\[
\langle \Sigma\cup\{\sharp\}, \{U_{c}\}_{c\in\Sigma\cup\{\sharp\} }, \{O_{c}\}_{c\in\Sigma\cup\{\sharp\}}, \pi_0, F \rangle 
\] such that every $U_{c}\in\mathbb{C}^{m\times m}$ is a unitary matrix, i.e. $U_{c}^{-1}=U_{c}^{\dag}$, and  
every $O_{c}\in\mathbb{C}^{m\times m}$ is a quantum observable, i.e. an Hermitian matrix. 
For each $O_{c}$ we denote the set of its orthogonal eigenspaces by $\{E_1, \ldots, E_{k} \}$. 
For each $E_i$ we denote by $P_i$ the corresponding orthogonal projection matrix. The initial state of the system is denoted by 
 $\pi_0\in\mathbb{C}^{1\times m}$. We denote by $F\subseteq V(O_{\#})$ the final accepting eigenvalues. 
\end{definition}
\begin{definition}[\lats\, language recognition]
Starting from the initial state $\pi_0$, at each step of the computation an input letter $c$ is read, 
then the unitary operator $U_c$ is applied and the system is measured by means of $O_c$. At the end of the input word $w$ we 
measure the system with the end-marker observable $O_{\#}$ and we probabilistically accept or reject $w$ according to whether this last measure is in $F$ or not. 
A \lat, recognizes languages by bounded (double-sided) 
error mode of acceptance \cite{Aal06}. More precisely, we say that the automaton $A$ recognizes the language $L$ with bounded 
(two-sided) error if $M$ accepts any $w\in L$ and rejects any $w\not\in L$ with probability at least $p$, where $p > \frac{1}{2}$.
Observe that bounded (two-sided) error mode of acceptance is equivalent to $\delta$-isolated $\lambda$ cut-point mode of acceptance, provided that $\lambda = \frac{1}{2}$.
\end{definition}
It has been proved in \cite{Aal06}, as a main result, the following characterization in terms of block-group syntactic monoids.
\begin{definition}[Block-group monoids]
Let $M$ be a finite monoid and let $L,R$ denote respectively the left and the right Green's equivalence relations.
We say that $M$ is a block-group if and only if every $R$-class and every $L$-class of $M$ contains at most one idempotent element. 
The class of block-group monoids is a pseudovariety of monoids, we denote it by $\bg$.
\end{definition}
\begin{theorem}[Ambainis et.al. \cite{Aal06}] \label{thm:ambainis}
Let $L\in\Sigma^*$ be a formal language over the alphabet $\Sigma$. Then $L$ is recognized by some 
\lat\, if and only if its syntactic monoid is a block-group.
\end{theorem}
\subsection{Outline of the main result}
We formalize our main result in Theorem \ref{th:MAIN}, in this section we also outline the core ideas of its proof.
\begin{theorem}[Algebraic Characterization of $\lmo(\Sigma)$ languages]\label{th:MAIN}
Let $L\in\Sigma^*$ be a formal language. Then the following four propositions are equivalent.
\begin{enumerate}
\renewcommand\labelenumi{\theenumi}
\renewcommand{\theenumi}{(\roman{enumi})}
\item $L$ is recognized with isolated cut-point by some \mon.
\item the syntactic monoid of $L$ is a $J$-trivial finite monoid and literally satisfies the idempotency pseudoidentity $x^2=x$. 
\item $L$ is a literal idempotent and piecewise testable language.
\item $L$ lies in the boolean closure of languages of the following form, for any $k\geq 0$, 
\[\Sigma^* a_1 \Sigma^* a_2 \Sigma^*\cdots \Sigma^* a_k \Sigma^*, \text{ for } a_1, a_2, \ldots, a_k\in\Sigma \text{ and } 
 a_i \neq a_{i+1} \text{ for every } 1 \leq i < k. \]
\end{enumerate}
\end{theorem}
In the subsequent sections, we shall prove Theorem \ref{th:MAIN} by carefully analyzing both the behavior of \mons\, recognizing languages over $\Sigma$ and the structure of their syntactic monoids. The outline of this proof is given below:
\begin{enumerate}
\renewcommand\labelenumi{\theenumi}
\renewcommand{\theenumi}{(\roman{enumi})}
\item we prove that \mons\, are a sub-family of \lats\,, which implies that if a language is recognized by some \mon\, then its syntactic monoid is a block-group \cite{Aal06} .
\item from the fact that reading twice the same symbol letter does not alter the system, we use some properties of syntactic monoids to conclude that if a language is recognized by some \mon\,, then it is a literal idempotent piecewise testable language.
\item we explicitly exhibit a family of \mons\, recognizing literal idempotent piecewise testable languages, thus proving that all such languages are indeed recognized by some \mon\,.
\end{enumerate}

\section{Preliminary analysis of \mons\, over $\Sigma$}
In this section we begin to analyze the model of \mons\, over finite non-empty set alphabets.
As a first instance of this, the following proposition expresses the probabilistic behavior of a \mon\, as a sum of square amplitudes.
\begin{proposition}\label{prop:sumofsquares}
Let $A := \langle \Sigma\cup \{\#\}, (O_c)_{c\in\Sigma\cup \{\#\}}, \pi_0, F \rangle$ 
be a \mon\, over $\Sigma$ such that, for every $c \in\Sigma$, $O_c$ admits the following spectral decomposition
$O_{c} = \sum_{j=1}^{k(c)} \lambda_{j} P_{j}(c)$.
Then, for every $w\in\Sigma^*$ such that $w:=w_1\cdots w_n$ and such that $k(w_i)$ is 
the cardinality of the spectrum of $O_{w_i}$ for each $1\leq i\leq n$, 
the probabilistic behavior $p_A(w)$ of $w$ w.r.t. $A$ is given by the following equation:
\[
p_A(w) = \sum_{r_j \in F}\sum_{j_1, \ldots, j_n} \big|\big| \pi_0 \cdot\Big(\prod_{i=1}^{n}P_{j_i}(w_i)\Big)\cdot P_{r_j}(\#)\big|\big|_2^2
\]
provided that, for every $1 \leq i \leq n$, $j_i$ goes from $1$ to $k(w_i)$ and provided that $r_j$ goes among 
the projectors of $O_{\#}$ that are in $F$ only.
\end{proposition}
Then we observe that \mons\, over $\Sigma$ satisfy similar properties already proved in \cite{BMP10} for \mons\, over partially commutative monoids with idempotent generators.
\begin{proposition}[Bertoni, et al. \cite{BMP10}]\label{prop:linearRep}
We let $A$ to be a \mon\, over $\Sigma$ and we let the formal power series $p_A(w):\Sigma^*\rightarrow [0,1]$ to represent the probability that $A$ accepts $w$, 
thus $p_A$ is the probabilistic event induced by $A$. 
We define $\pmo(\Sigma)$ as the class of probabilistic events induced by \mons\, over $\Sigma$. Then the following propositions holds.
\begin{enumerate}
\renewcommand\labelenumi{\theenumi}
\renewcommand{\theenumi}{(\roman{enumi})}
\item The formal power series generated by a \mon\, $A$ on $\Sigma$ with $m$ states admits a linear representation $\langle \xi, (P(c))_{c\in\Sigma}, \eta \rangle$ where 
$\xi\in\mathbb{C}^{1\times m}$ and $||\xi|| = 1$, $P(c)\in\mathbb{C}^{m\times m}$ is a projector for all $c\in\Sigma$, and $\eta\in\mathbb{C}^{1\times m}$ satisfies $||\eta||\leq\sqrt{m}$.
\item \label{prop:closureHadamard} The class $\pmo(\Sigma)$ is closed under the operations of Hadamard product and $f$-complement;  
where for any formal series $\varphi, \psi:\Sigma^*\rightarrow \mathbb{C}$ and for any $w\in\Sigma^*$ 
the Hadamard product is defined as $(\varphi \odot \psi)(w) := \varphi(w)\cdot \psi(w)$ and the f-complement as $\varphi^{\text{f-c}}(w):=1-\varphi(w)$.
\item Let $L\in\lmo(\Sigma)$ and $\chi_L$ be its characteristic function. For any $\epsilon > 0$, there exists $\phi\in\pmo(\Sigma)$ such that $|\phi(w)-\chi_L(w)|<\epsilon$, for all $w\in\Sigma^*$.
\item Given a formal series $\phi:\Sigma^*\rightarrow [0,1]$, let $\langle \xi, (P(c))_{c\in\Sigma}, \eta \rangle$ be a linear representation of $\phi$, where $||\xi|| = 1$ and $P(c)$ is a projector for $c\in\Sigma$, $P(\epsilon)=I$. Suppose that $|\phi(w)-\lambda|\geq \delta > 0$ for all $w\in\Sigma^*$, and let $L:=\{w\in\Sigma^* | \phi(w)>\lambda\}$. Then $L$ is a regular language of $\Sigma^*$. Moreover, there exists a finite-state automaton $\langle \Sigma, Q, (\underline{\delta}_c)_{c\in\Sigma}, q_0, F \rangle$ recognizing $L$ such that, for any $w,u\in\Sigma^*$, the following holds: 
$\underline{\delta}_w(q_0)\neq \underline{\delta}_u(q_0)$ implies $||\xi P(w) - \xi P(u)||\geq \frac{\delta}{|| \eta ||}$.
\item Let $\Sigma$ be a finite non-empty set. The class $\lmo(\Sigma)$ is closed under intersection, complement and union.
\item $\lmo(\Sigma)$ is a boolean algebra of regular languages in $\Sigma^*$ with finite variation. In particular, if $L$ is a 
language recognized by a \mon\, over $\Sigma$ with $m$ states and isolation $\delta$, then 
$\sup_{x\in\Sigma^*}\text{var}_L(x)\leq \frac{m}{\delta^2}$.
\end{enumerate}
\end{proposition}
In the next proposition we observe further that any convex linear combination of \mons\, over $\Sigma$ induce a convex linear combination of probabilistic events.
\begin{proposition}
\label{prop:combinazionelineare}
Let $A := \langle \Sigma\cup \{\#\}, (O^A_c)_{c\in\Sigma}, \pi^A, F^A \rangle$
 and $B := \langle \Sigma\cup \{\#\}, (O^B_c)_{c\in\Sigma}, \pi^B, F^B \rangle$ be two \mons .
Let $\alpha, \beta\in\mathbb{R}$ be non-negative real numbers such that $\alpha+\beta=1$.
We define the convex linear combination $A\oplus B$ of $A$ and $B$ w.r.t. $\alpha$ and $\beta$ as follows, 
\[
A\oplus B := \left\langle \Sigma\cup\{\#\}, (O^A_c)_{c\in \Sigma\cup\{\#\}}\oplus (O^B_c)_{c\in \Sigma\cup\{\#\}}, 
( \sqrt{\alpha}\cdot\pi^A)\oplus (\sqrt{\beta}\cdot\pi^B), F^A\cup F^B \right\rangle
\] Then $p_{A\oplus B}(x) = \alpha\cdot p_A(x)  + \beta\cdot p_B(x)$ holds for every $x\in\Sigma^*$.
\end{proposition}




By using proposition \ref{prop:combinazionelineare} it is possible to prove the following Proposition.
\begin{proposition} \label{prop:taglioImmateriale}
Let $A$ be a \mon\, over $\Sigma$ and let 
 $L_A$ be the language recognized by $A$ with cut-point  
$\lambda_A\neq 1/2$ isolated by $\delta > 0$. Then there exists a \mon\, $A'$ over $\Sigma$ 
recognizing $L_A'=L_A$ with cut-point $\lambda_{A'}=1/2$ isolated by $\delta_{A'} = \frac{\delta}{2\cdot\text{max}\{\lambda, 1-\lambda\}}$.
\end{proposition}
To conclude this section, we link the family on \mons\, with the Latvian model introduced by Ambainis, et.al \cite{Aal06}.
We will show that the \mon\, model is in fact a subclass of the Latvian one. 
The following proposition shows that every \mon\, over $\Sigma$ is in fact a \lat\, with identical unitary evolution.
\begin{proposition}\label{prop:lmoambainis}
The class of \mon\ over $\Sigma$ is a sub-class of \lat. In particular, let $A$ be a \mon\ on $\Sigma$ and let 
$L_A$ be a language recognized by $A$ with cut-point $\lambda$ isolated by $\delta$. 
Then there exists a \lat\, $A'$ recognizing $L_{A'}=L_{A}$ with bounded (two-sided) error 
$p = \frac{1}{2} + \frac{\delta}{2\text{max}\{\lambda, 1-\lambda\}}$; 
that is to say with cut-point $\lambda' = \frac{1}{2}$ isolated by 
$\delta' = \frac{\delta}{2\cdot \text{max}\{\lambda, 1-\lambda\}}$.
\end{proposition}\begin{proof}
Let $A$ be a \mon\, over $\Sigma$ recognizing $L_A$ with cut-point $\lambda$ isolated by $\delta > 0$. 
By Proposition \ref{prop:taglioImmateriale}, there exists a \mon\, over $\Sigma$, denoted by 
$A'=\langle \Sigma\cup\{\sharp\}, (O'_c)_{c\in\Sigma\cup\{\sharp\}}, \pi'_0, F' \rangle$, recognizing  
$L_{A'}=L_{A}$ with cut-point $\frac{1}{2}$ isolated by 
$\frac{\delta}{2\cdot\text{max}\{\lambda, 1-\lambda\}}$. 
Consider the Latvian automaton $B:=\langle \Sigma\cup\{\sharp\}, (U_c), O'_c, F' \rangle$ defined by taking 
$U_{c}=I$ for every $c\in\Sigma$, where $I$ is the identity matrix. Then $L_B=L_A'=L_A$ is recognized by $B$ with 
bounded (two-sided) error given by $p = \frac{1}{2} + \frac{\delta}{2\cdot\text{max}\{\lambda, 1-\lambda\}}$.
\end{proof}
\section{On syntactic monoids of $\lmo(\Sigma)$ languages}
In this section we study syntactic monoids of $\lmo(\Sigma)$ languages. Observe that by Theorem \ref{thm:ambainis} and by 
Proposition \ref{prop:lmoambainis}, if $L\in\lmo(\Sigma)$ then its syntactic monoid is a block-group. That is to say, by Proposition \ref{prop:linearRep}, that 
$\lmo(\Sigma)\subseteq V_{\Sigma}(\bg)$ as a sub-boolean algebra. \\\\In the next proposition we characterize the finite variation property in terms of properties satisfied by syntactic monoids of finite variation languages.
By directly proving closure properties, it is not difficult to show that the class of finite variation regular languages is a 
$*$-variety of Eilenberg. In fact the following proposition holds true.
\begin{proposition}\label{prop:FVRtriviale}
Let $L$ be a regular language. The following propositions are equivalent.
\begin{enumerate}
\renewcommand\labelenumi{\theenumi}
\renewcommand{\theenumi}{(\roman{enumi})}
\item $L$ has finite variation. 
\item Every strongly connected component of the minimum automaton recognizing $L$ has one vertex only.
\item There exists a total ordering on the set $Q$ of the states of the minimum automaton 
recognizing $L$ such that $qa\geq q$ for every $q\in Q$ and every $a\in\Sigma$.
\item The syntactic monoid $M$ of $L$ is an $R$-trivial monoid.
\end{enumerate}
\end{proposition}
\begin{proof}
$\text{(i)}\iff \text{(ii)}$ 
Consider the minimum automaton recognizing $L$ and assume 
it admits a strongly connected component with more than one vertex, 
then there exists an oriented simple cycle $p$ visiting at least two vertices. 
If $w\in\Sigma^*$ is the label word of $p$, then $\{w^n\}_{n\geq 1}$ shows that the variation is not finite.
Vice-versa, assume that every strongly connected component is trivial. Then 
the variation is upper-bounded by the longest path with no self-loops and is therefore finite.\\
$\text{(ii)}\Rightarrow \text{(iii)}$
If we remove every self-loop from a graph we obtain an acyclic oriented graph.
Consider the graph induced by the minimum automaton recognizing $L$ and 
order the nodes of this graph by the topological ordering relation, then (iii) holds.\\
$\text{(iii)}\Rightarrow \text{(iv)}$ 
By hypothesis $qf\geq q$ for every $f\in\Sigma^*$ and $q\in Q$. Let $M$ be the syntactic monoid of $L$.
Assume that $u$ and $v$ generate the same right ideal of $M$, that is to say $uM=vM$. 
Then $u=vx$ and $v=uy$ for some $x,y\in \Sigma^*$. Hence for every $q$ it holds 
$qu=qvx\geq qv=quy\geq qu$. That is $qu=qv$. It follows that $u=v$.\\
$\text{(iv)}\Rightarrow \text{(ii)}$
Assume that $q$ and $q'$ are in the same strongly connected component.
Then, there exist two elements $u, v\in M$ such that 
$qu=q'$ and $q'v=q$. This imply $q(uv)^n=q$ and $q(uv)^nu=q'$ for every 
$n\in\mathbb{N}$. Pick an integer $n$ such that $(uv)^n$ is idempotent. 
Then $(uv)^n$ generate the same right ideal of $(uv)^n u$, since $(uv)^n\in ((uv)^n u) M$. 
This imply that $(uv)^n=(uv)^n u$, since $M$ is $R$-trivial by hypothesis. From this, it follows $q=q'$.
\end{proof}
\begin{corollary}
Let $L$ be a regular language. Then $L$ has finite variation if and only if its syntactic monoid is $\R$-trivial.
\end{corollary}
As a direct consequence, we have proved the following.
\begin{proposition}
Let $L\in\lmo(\Sigma)$ be a language recognized by some \mon\ with isolated cutpoint. Then its syntactic monoid $M(L)$ 
is an $R$-trivial block-group. Formally $ M(L) \in \bg\cap \R$ and $\lmo(\Sigma)\subseteq V_{\Sigma}(\bg\cap\R)$ as a
sub-boolean algebra.
\end{proposition}
We observe further that the following holds true.
\begin{proposition}\label{prop:rjtrivial}
Every $R$-trivial block-group monoid is $J$-trivial.
\end{proposition}
\begin{proof} 
If $S$ is a semigroup and $a\in S$, we say that $a$ is regular if there exists $a\in S$ such that $asa=a$. 
We say that an $L,R,J$ class is regular if all its elements are regular.
If $M$ is an $R$-trivial monoid then every regular $J$-class is a regular $L$-class $C$, such that every element of $C$ is 
idempotent. The block-group condition imply that $C$ has at most one idempotent. Then, $C$ has only one element. Then, every 
regular $J$-class of $M$ is trivial and this imply that every $J$-class of $M$ is trivial (see \cite{Pi86}, pag. 65, Proposition 4.1). 
That is to say that $M$ is a $J$-trivial monoid.
\end{proof}
Thus, by proposition \ref{prop:rjtrivial}, we have also proved the following.
\begin{proposition}
Let $L\in\lmo(\Sigma)$ be a language recognized by some \mon\ with isolated cutpoint. Then its syntactic monoid $M(L)$ 
is $J$-trivial. Formally $ M(L) \in \J$ and $\lmo(\Sigma)\subseteq V_{\Sigma}(\J)$ as a sub-boolean algebra.
\end{proposition}
To conclude, we summarize the contents of this section in the following Proposition.
\begin{proposition}\label{prop:forwardDirection}
Let $L\in\lmo(\Sigma)$ be a language recognized by some \mon\ with isolated cutpoint. Then the following propositions holds.
\begin{enumerate}
\renewcommand\labelenumi{\theenumi}
\renewcommand{\theenumi}{(\roman{enumi})}
\item $L$ is a picewise testable language, $L\in \pt$ or equivalently $L\in V_{\Sigma}(\J)$.
\item $L$ is literally idempotent $L\in \id$. Thus $L\in \pt\cap\id = V_{\Sigma}(\J)\cap \id$.
\item The syntactic monoid of $L$ is $J$-trivial and literally satisfies the pseudoidentity $x^2=x$. 
That is to say $L\in V_{\Sigma}(\overline{\J})$.
\item $L$ lies in \lipt, that is the boolean closure of languages of the following form, for any $k\geq 0$, \[\Sigma^* a_1 \Sigma^* a_2 \Sigma^*\cdots \Sigma^* a_k \Sigma^*, \text{ for } a_1, a_2, \ldots, a_k\in\Sigma \text{ and } 
 a_i \neq a_{i+1} \text{ for every } 1 \leq i < k. \]
\end{enumerate}
\end{proposition}
\begin{proof}
Item (i) is exactly Proposition \ref{prop:rjtrivial}. Item (ii) follows from propositon \ref{prop:rjtrivial} and the fact that projection 
operators $P_i$ of any \mon\,'s observable are orthogonal and idempotent by definition. Items (iii) and (iv) follow from Theorem 
\ref{thm:klimapolak}, which is due to Kl\'ima and Pol\'ak \cite{KP08}.
\end{proof}

\section{\mons\, recognizing literally idempotent piecewise testable languages}
In this section we show how languages in \lipt\ can be recognized by \mons, proving in fact the converse of Proposition \ref{prop:forwardDirection}.\\\\
In this section it is convenient to define the language 
$L[a_1, \ldots, a_k] := \Sigma^*a_1\Sigma^*\cdots\Sigma^*a_k\Sigma^*$,
for any $k\geq 1$ and $a_1,\ldots, a_k \in\Sigma$, such that $a_i \neq a_{i+1}$ for each $1\leq i< k$. 
We also let $S := \{a_1, \ldots, a_k\}$.
\begin{definition}[On $j_i^{(\alpha)}$ indexes]\label{def:jindici}
Let $L[a_1, \ldots, a_k] := \Sigma^*a_1\Sigma^*\cdots\Sigma^*a_k\Sigma^*$,
for any $k\geq 1$ and $a_1,\ldots, a_k \in\Sigma$, such that $a_i \neq a_{i+1}$ for each $1\leq i< k$. 
Let $S := \{a_1, \ldots, a_k\}$.
For every $\alpha \in S$, let $\#\alpha$ be the number of times that $\alpha$ appears as a letter in the word $a_1a_2\cdots a_k$.
Let \[j^{(\alpha)}_1 < j^{(\alpha)}_2 < \cdots < j^{(\alpha)}_{\#\alpha} \text{ be all the indexes such that }
\alpha = a_{j^{(\alpha)}_1} = \ldots = a_{j^{(\alpha)}_{\#\alpha}}\] in increasing order. For each $1\leq i\leq \#\alpha$ it holds 
that $1 \leq  j_i^{(\alpha)} \leq k$. 
\end{definition}

\begin{definition}\label{operatori:compatta}
We define, for every $\alpha\in S$ and $k\geq 1$, two orthogonal projectors of dimension $(k+1)\times (k+1)$: the up operator $P^{(k)}_{\nearrow}(\alpha)$ and the down operator $P^{(k)}_{\searrow}(\alpha)$, such that
\[
\left(P^{(k)}_{\nearrow}(\alpha)\right)_{rs} =
\left\{
	\begin{array}{cl}
	1 & \mbox{ if } r=s \mbox{ and } \forall\, 1\leq i\leq \#\alpha \mbox{ it holds } r,s\notin\{j_{i}^{(\alpha)}, j_{i}^{(\alpha)}+1\},\\
	\frac{1}{2} & \mbox{ if } \exists\, 1\leq i\leq \#\alpha \mbox{ such that } r,s\in\{j_{i}^{(\alpha)}, j_{i}^{(\alpha)}+1\},\\
	0 & \mbox{ otherwise,}
	\end{array}
\right.
\]
\[
\left(P^{(k)}_{\searrow}(\alpha)\right)_{rs} =
\left\{
	\begin{array}{cl}
	\frac{1}{2} & \mbox{ if } r=s \mbox{ and } \exists\, 1\leq i\leq \#\alpha \mbox{ such that } r,s\in\{j_{i}^{(\alpha)}, j_{i}^{(\alpha)}+1\},\\
	-\frac{1}{2} & \mbox{ if } r\neq s \mbox{ and } \exists\, 1\leq i\leq \#\alpha \mbox{ such that } r,s\in\{j_{i}^{(\alpha)}, j_{i}^{(\alpha)}+1\},\\
	0 & \mbox{ otherwise.}
	\end{array}
\right.
\]
\end{definition}
\begin{definition}\label{operatori:estesa}
We also give an extended representation of $P^{(k)}_{\nearrow}$ and $P^{(k)}_{\searrow}$, 
which is equivalent to Definition \ref{operatori:compatta}.
Let $u:=(+\frac{1}{\sqrt{2}}, +\frac{1}{\sqrt{2}})$ and $d:=(+\frac{1}{\sqrt{2}}, -\frac{1}{\sqrt{2}})$. Consider the following two elementary operators of orthogonal projection in dimension $2\times 2$, 
the up-diagonal operator $P_{\nearrow}$ and the down-diagonal operator $P_{\searrow}$, defined as 
\[
P_{\nearrow} :=
u^T u =
\left( 
\begin{array}{c c}
+\frac{1}{2} & +\frac{1}{2} \\
+\frac{1}{2} & +\frac{1}{2} \\
\end{array} 
\right)
\]
\[
P_{\searrow} := 
d^T d=
\left( \begin{array}{c c}
+\frac{1}{2} & -\frac{1}{2} \\
-\frac{1}{2} & +\frac{1}{2} \\
\end{array} \right)
\]
Then, for $k>1$, the up-diagonal operator 
$P^{(k)}_{\nearrow}(\alpha) \in\mathbb{C}^{(k+1)\times (k+1)}  \text{ is defined as follows } $
\[ \scriptsize
\begin{array}{c | c | c c c | c | c c c | c | c c c | c |c c c | c | c}
& \mbox{\fontsize{8}{8}\selectfont $1,\ldots$}
&\mbox{\fontsize{8}{8}\selectfont $j^{(\alpha)}_1$}
&\mbox{\fontsize{8}{8}\selectfont $j^{(\alpha)}_1+1$} & &
&\mbox{\fontsize{8}{8}\selectfont $j^{(\alpha)}_2$} &
\mbox{\fontsize{8}{8}\selectfont $j^{(\alpha)}_2+1$}& & & &\cdots &
&&\mbox{\fontsize{8}{8}\selectfont
  $j^{(\alpha)}_{\#\alpha}$}&\mbox{\fontsize{8}{8}\selectfont
  $j^{(\alpha)}_{\#\alpha}+1$}&&\mbox{\fontsize{8}{8}\selectfont
  $\ldots, k+1$}&\\
\hline
\mbox{\fontsize{8}{8}\selectfont $1,\ldots$}& I & & & & & & & & & & & &&&&&&\\
\hline
\mbox{\fontsize{8}{8}\selectfont $j^{(\alpha)}_1$}& &   &  &  & & & & & & & & &&&&&&\\ 
         & &   & P_{\nearrow} &  & & & & & & & & &&&&&&\\ 
\mbox{\fontsize{8}{8}\selectfont $j^{(\alpha)}_1+1$}& &   &   & & & & & & & & &&&&&&\\
\hline
& &       &     &   & I & & & & & & &&&&&&\\
\hline
\mbox{\fontsize{8}{8}\selectfont $j^{(\alpha)}_2$}& &   &  &  & & & & & && &&&&&&\\ 
         & &   & &  & & & P_{\nearrow} & & & & & &&&&&&\\ 
\mbox{\fontsize{8}{8}\selectfont $j^{(\alpha)}_2+1$}& &   &   & & & & & & & &&&&&&&&\\
\hline
         & &   & &  & & &  & & I & & &&&&&&\\ 
\hline
& &   &  &  & & & & & & & & &&&&&\\ 
\vdots  & &   & &  & & & & & & & \ddots &&&&&&\\ 
& &   &   & & & & & & & &&&&&&&\\
\hline
   & &       &     &   & & & & & & & &&I&&&&\\
\hline
\mbox{\fontsize{8}{8}\selectfont $j^{(\alpha)}_{\#\alpha}$}  
& & & &   &  &  & & & & & & & & &       & &\\ 
& & & &   &  &  & & & & & & & & & P_{\nearrow} & &\\ 
\mbox{\fontsize{8}{8}\selectfont $j^{(\alpha)}_{\#\alpha+1}$}  
& & & & & & & & & & & & & & & & & &\\
\hline
\mbox{\fontsize{8}{8}\selectfont $\ldots, k+1$} 
& & & & & & & & & & & & & & & & &I &\\ 
\hline
\end{array}\]\\ in which not appearing elements are intended to be null, and every $I$ is an identity matrix. \\
Next, for $k>1$, we define
$ P^{(k)}_{\searrow}(\alpha)\in \mathbb{C}^{(k+1)\times (k+1)} \text{ as follows} $
\[ \scriptsize
\begin{array}{c | c | c c c | c | c c c | c | c c c | c|ccc|c|c}
& \mbox{\fontsize{8}{8}\selectfont $1,\ldots$} &\mbox{\fontsize{8}{8}\selectfont $j^{(\alpha)}_1$} &\mbox{\fontsize{8}{8}\selectfont $j^{(\alpha)}_1+1$} & & &\mbox{\fontsize{8}{8}\selectfont $j^{(\alpha)}_2$} & \mbox{\fontsize{8}{8}\selectfont $j^{(\alpha)}_2+1$}& & & &\cdots & & &\mbox{\fontsize{8}{8}\selectfont $j^{(\alpha)}_{\#\alpha}$}&\mbox{\fontsize{8}{8}\selectfont $j^{(\alpha)}_{\#\alpha+1}$}&&\mbox{\fontsize{8}{8}\selectfont $\ldots, k+1$}&\\
\hline
\mbox{\fontsize{8}{8}\selectfont $1, \ldots$}& \mathbf{0} & & & & & & & & & & & &&&&&&\\
\hline
\mbox{\fontsize{8}{8}\selectfont $j^{(\alpha)}_1$}& &   &  &  & & & & & & & & &&&&&&\\ 
         & &   & P_{\searrow} &  & & & & & & & & &&&&&&\\ 
\mbox{\fontsize{8}{8}\selectfont $j^{(\alpha)}_1+1$}& &   &   & & & & & & & & &&&&&&\\
\hline
& &       &     &   & \mathbf{0} & & & & & & &&&&&&\\
\hline
\mbox{\fontsize{8}{8}\selectfont $j^{(\alpha)}_2$}& &   &  &  & & & & & & & &&&&&&\\ 
         & &   & &  & & & P_{\searrow} & & & & & &&&&&\\ 
\mbox{\fontsize{8}{8}\selectfont $j^{(\alpha)}_2+1$}& &   &   & & & & & & & &&&&&&&\\
\hline
         & &   & &  & & &  & & \mathbf{0} & & &&&&&&\\ 
\hline
& &   &  &  & & & & & & & & &&&&&\\ 
\vdots  & &   & &  & & & & & & & \ddots &&&&&&\\ 
& &   &  & & & & & & & & &&&&&&\\
\hline
   & &       &     &   & & & & & & & && \mathbf{0} &&&&\\
\hline
\mbox{\fontsize{8}{8}\selectfont $j^{(\alpha)}_{\#\alpha}$}  & & & &   &  &  & & & & & & & &&&&\\ 
         & & & &   & &  & & & & & & & & & P_{\searrow}&&\\ 
\mbox{\fontsize{8}{8}\selectfont $j^{(\alpha)}_{\#\alpha+1}$}  & & & &   & &  & & & & & & & & & &  &\\
\hline
\mbox{\fontsize{8}{8}\selectfont
  $\ldots, k+1$} & & & &   & &  & & & & & & & & & & & \mathbf{0}\\ 
\hline
\end{array}
\]\\ in which not appearing elements are intended to be null and the $\mathbf{0}$ entries are block matrices with null elements.
\end{definition}

\begin{definition}\label{def:automatapicewise}
Let $S:=\{a_1, \ldots, a_k\}$ for some $a_i\in\Sigma$. By calling $e_j$ the boolean row vector such that 
$(e_j)_i=1\Leftrightarrow i=j$, we define $A[a_1,\ldots,a_k]=\langle \Sigma\cup\{\#\}, \pi_0^{(k)}, \{O_c^{(k)}\}_{c\in\Sigma\cup\{\#\}}, F^{(k)} \rangle$ as the \mon\ where
\begin{itemize}
	\item $\pi_0^{(k)}=e_1\in\C^{1\times(k+1)}$,
	\item for $\alpha\in S$, the associated projectors of $O_\alpha^{(k)}$ are $P^{(k)}_{\nearrow}(\alpha)$ and $P^{(k)}_{\searrow}(\alpha)$,
	\item with each $O_c^{(k)}$ such that $c\in\Sigma\setminus S$, we associate the identity matrix $I_{(k+1)\times(k+1)}$,
	\item the projector of the accepting result of $O_\#^{(k)}$ is $(e_{k+1})^Te_{k+1}$, i.e. the $(k+1)\times(k+1)$ boolean matrix having a 1 only in the bottom right entry. We denote it by $P^{(k)}_{\text{acc}}$.
\end{itemize}
\end{definition}
We begin a careful analysis of the computing behavior of $A[a_1, \ldots, a_k]$ as defined in Definition 
\ref{def:automatapicewise}. Observe that $P^{(k)}_{\nearrow}(\alpha)$ maps contiguous 
pairs of coordinates of $\mathbf{x}$ into their mid-points according to the indexes $\{j^{(\alpha)}_q\}_{q=1}^{\#\alpha}$. Clearly, if we start with the vector $\mathbf{x}= (1, 0 \dots 0)$, 
each reading of a letter $a_i$ will average $x_i$ and $x_{i+1}$, so that if we read all letters in the
good order, the last coordinate will be nonzero, but would be otherwise. This, in turn, motivates the following Lemmata.
\begin{lemma}\label{lemma:midpoints}
Consider the \mon\, $A[a_1, \ldots, a_k]$ with up-diagonal projectors $\{P_{\nearrow}^{(k)}(\alpha)\}_{\alpha\in S}$, for $S:=\{a_1, \ldots, a_k\}$. Let $\mathbf{x}=(x_1, x_2, \ldots, x_{k+1})\in\mathbb{R}^{1\times (k+1)}$ 
be an $n$-dimensional real vector. Let $j^{(\alpha)}_1 < \cdots < j^{(\alpha)}_{\sharp\alpha}$ be the indexes as in Definition \ref{def:jindici}. Let $\mathbf{y}:=\mathbf{x}P^{(k)}_{\nearrow}(\alpha)$.
Then \[(\mathbf{y})_i =  
\left\{
\begin{array}{c l}
(\mathbf{y})_i & \text{ if } i\not\in\{ j^{(\alpha)}_1,\ldots, j^{(\alpha)}_{\#\alpha}\}\cup 
\{ j^{(\alpha)}_1+1,\ldots, j^{(\alpha)}_{\#\alpha}+1\}\\

\frac{(\mathbf{y})_i + (\mathbf{y})_{i+1}}{2} & \text{ if } i\in\{ j^{(\alpha)}_1,\ldots, j^{(\alpha)}_{\#\alpha}\} \\

\frac{(\mathbf{y})_{i-1} + (\mathbf{y})_i}{2} & \text{ if } i\in\{ j^{(\alpha)}_1+1,\ldots, j^{(\alpha)}_{\#\alpha}+1\} \\
\end{array}
\right.
\]
\end{lemma}
\begin{lemma}\label{lemma:ampiezzapositivalimitata}
Consider the language $L[a_1, \ldots, a_k]$ over the alphabet $\Sigma$ and let $S:=\{a_1, \ldots, a_k\}$.
Let $w := w_1\cdots w_n$ be a word of $\Sigma^*$. Consider the class of up-diagonal orthogonal projectors 
$\{P_{\nearrow}^{(k)}(a_i)\}_i$ associated to the \mon\, $A[a_1, \ldots, a_k]$, with initial quantum state 
$\pi_0:=e_1 \in \mathbb{C}^{1\times (k+1)}$. 
Let us define \[\pi := \pi_0\cdot \prod_{i=1}^{n}P^{(k)}_{\nearrow}(w_i)\] 
where, for notational convenience, $P_{\nearrow}^{(k)}(w_i)=I$ if $w_i\in\Sigma\setminus S$.
Then the following two propositions holds: 
\begin{enumerate}
\item if $w\in L[a_1, \ldots, a_k]$ then, for every $1 \leq i \leq k+1$, $(\pi)_{i}>0$.
\item for every $1 \leq j\leq i \leq k+1$, it holds that: \[\text{if } \left(\pi\right)_i > 0 \text{ then } \left(\pi\right)_{j} \geq 2^{-k}\]
\end{enumerate}
\end{lemma}
\begin{proof} 
Observe that by Lemma \ref{lemma:midpoints} any projector $P_{\nearrow}^{(k)}(\alpha)$ sends contiguous pairs of coordinates into their mid-point according to indexes $\{j^{(\alpha)}_q\}_{q=1}^{\#\alpha}$ and leaves the others coordinates intact. \\
We prove (i). If $w\in L[a_1, \ldots, a_k]$ then $w = u_1 a_1 u_2 a_2 \cdots u_{k} a_k u_{k+1}$ for some $u_1, \ldots, u_{k+1}\in\Sigma^*$.
As we start with $e_1$, by applying a projector $P^{(k)}_{\nearrow}(\alpha)$ to $e_1$ as defined in Definition \ref{operatori:estesa}, we perturb the second coordinate from a null to a non null value if and only if $j^{(\alpha)}_1 = 1$, that is if and only if $\alpha = a_1$; otherwise we leave $e_1$ intact. Once we have applied $P^{(k)}_{\nearrow}(a_1)$, we perturb the third coordinate of the 
vector state if and only if $j^{(\alpha)}_1 = 2$, that is if and only if $\alpha = a_2$; otherwise, by applying any other up-diagonal projector we leave perturbed coordinates with a non null value. Iterating this way, once we have applied 
$P^{(k)}_{\nearrow}(a_i)$ we have perturbed the $i+1$-th coordinate, and we perturb the $i+2$-th coordinate if and 
only if we apply $P^{(k)}_{\nearrow}(a_{i+1})$, leaving the previously perturbed coordinates with a non null value otherwise. This imply the thesis.\\
Now we prove (ii). For any $k\in\mathbb{N}_0$, any $\mathbf{x}\in\mathbb{R}^{1\times (k+1)}$ and each $1\leq i \leq k$ let $T_i:\mathbb{R}^{1\times (k+1)}\rightarrow\mathbb{R}^{1 \times (k+1)}$ be the linear transformation defined by  
 \[
 (T_i(\mathbf{x}))_j := 
 \left\{
 \begin{array}{c l}
  (\mathbf{x})_j & \text{ if } j\not\in\{i, i+1\}\\
\frac{(\mathbf{x})_j + (\mathbf{x})_{j+1}}{2} & \text{ if } j = i \\
\frac{(\mathbf{x})_{j-1} + (\mathbf{x})_{j}}{2} & \text{ if } j = i+1 \\
 \end{array}
 \right.
 \] \\
By induction on $k$, it follows that for any length $n\in\mathbb{N}_0$, any choice of indexes $1\leq i_1, \ldots, i_n\leq k$ and any $1\leq j\leq i\leq k+1$, 
 it holds, provided we define $T_{[n]}(e_1):=T_{i_n}\circ \cdots \circ T_{i_1} (e_1)$, that 
 \[ 
 \text{ if } \big(T_{[n]}(e_1)\big)_i > 0 \text{ then } \big(T_{[n]}(e_1)\big)_j \geq 
 2^{-k}
 \]
Thus, for every $1\leq i \leq k+1$ we have $\big(T_{[m+1]}(e_1)\big)_{i}\geq 2^{-k}$. Subsequent applications of $T$ can not 
decrease this bound, since $T$ replace contiguous coordinates pairs by their mid-point.
\end{proof}
\begin{lemma}\label{lemma:nonperturbazione}
Let $i\in\mathbb{N}$ be a non-negative integer. 
For any integer $k>i$ and any $S:=\{a_1, \ldots, a_k\}$ consider the \mon\, $A:=A[a_1, \ldots, a_k]$ and its associated projectors. 
Let $w\in L[a_1, \ldots, a_i]\setminus L[a_1, \ldots, a_{i+1}]$ be of length $|w|=n$. 
Then for every $d_1, d_2, \ldots, d_n\in\{\nearrow, \searrow\}$ and 
for every $m\in\mathbb{N}_0$ such that $i+2 \leq m \leq k+1$ it holds that 
\[
\left(\pi_0\cdot\prod_{j=1}^{n} P^{(k)}_{d_j}(w_j)\right)_{m} = 0
\]
Where, for notational convenience, $P_{\searrow}^{(k)}(c) = P_{\nearrow}^{(k)}(c) = I$ if $c\in\Sigma\setminus S$  
and $I$ is the identity matrix of size $(k+1)\times (k+1)$.
\end{lemma}
\begin{proof}
For every $x\in\Sigma^*$ of length $|x|$ and every $\mathbf{d}=(d_1, \ldots, d_{|x|})\in\{\searrow, \nearrow\}^{|x|}$, we shall denote $P_{\mathbf{d}}(x):=P^{(k)}_{d_1}(x_1)\cdots P^{(k)}_{d_{|x|}}(x_{|x|})$.
We proceed by induction on $i\in\mathbb{N}$. 
If $i=0$, then $w\in \Sigma^*\setminus L[a_1]$ imply that $w$ does not contains the letter $a_1$. Then, because of 
Definition \ref{operatori:estesa}, for every $\mathbf{d}$, it holds that 
\[\text{ if } q\geq 2 \text{ then } (\pi_0 P_{\mathbf{d}}(w))_q=0\] since 
the projectors $P_{\nearrow}(a_1), P_{\searrow}(a_1)$ of $a_1$ are never applied and they are the only projectors of $A$ which can perturb the second coordinate from a null to a non non-null value. This proves the base case of induction.\\\\
Let us suppose that the thesis is true for $i\geq 0$, we analyze the case $i+1$. 
Then $w\in L[a_1, \ldots, a_{i+1} ]\setminus L[a_1, \ldots, a_{i+2}]$ so there exist words $w', w''$ such that 
$w=w'a_{i+1}w''$ where $w'\in L[a_1, \ldots, a_{i}]\setminus L[a_1, \ldots, a_{i+1}]$ and $w''$ does not contains $a_{i+2}$.
By induction hypothesis, for every $\mathbf{d}$ and $k>i$ it holds that 
\[
\text{ if } i+2 \leq q \leq k+1 \text{ then }
\left(\pi_0 P^{(k)}_{\mathbf{d}}(w')\right)_{q} = 0
\]
Then for every $\mathbf{d}$ and $d\in\{\searrow, \nearrow\}$ it holds that 
\[\text{ if } i+3 \leq q \leq k+1 \text{ then } 
\left(\pi_0 P^{(k)}_{\mathbf{d}}(w')P^{(k)}_d(a_{i+1})\right)_{q} = 0\] 
since the projector $P^{(k)}_{d}(a_{i+1})$ disrupt, at most, the {i+2} coordinate 
$(\pi_0 P^{(k)}_{\mathbf{d}}(w')P^{(k)}_{d}(a_{i+1}))_{i+2}$.\\
Since $w''$ does not contains the letter $a_2$, then for every $\mathbf{d'}, \mathbf{d''}$ it holds that 
\[
\text{ if } i+3 \leq q \leq k+1 \text{ then } 
\left(\pi_0 P^{(k)}_{\mathbf{d'}}(w')P^{(k)}_d(a_{i+1})P^{(k)}_{\mathbf{d''}}(w'') \right)_{q} = 0
\]
This concludes the proof.
\end{proof}
Finally, we are in the position to characterize $L[a_1, \ldots, a_k]$ language recognition by $A[a_1, \ldots, a_k]$ automata. We also observe that this generalizes Example 3 in \cite{BMP10}.
\begin{theorem}\label{th:PTtoMON}
The automaton $A[a_1, \ldots, a_k]$ recognizes $L[a_1, \ldots,
a_k]$ with cutpoint $\lambda = \frac{1}{2^{2k+1}}$ isolated by 
$\delta = \frac{1}{2^{2(k+1)}}$.
\end{theorem}
\begin{proof} We seek a lower-bound for $p_A(w)$, for any $w\in\Sigma^*$. 
We organize the proof in two cases.\newline
\emph{Case 1: $w\in L[a_1, \ldots, a_k]$}. 
From Lemma \ref{lemma:ampiezzapositivalimitata}, item (i), it follows that 
\[\text{ if } w \in L[a_1, \ldots, a_k] \text{ then } \Big(\pi_0 \cdot\prod_{i=1}^{n}P^{(k)}_{\nearrow}(w_i) \Big)_{k+1} > 0 \]  
We also recall that, according to Proposition \ref{prop:sumofsquares}, we have:
\[
p_A(w) = \sum_{r_j \in F}\sum_{j_1, \ldots, j_n} \big|\big| \pi_0 \cdot \Big(\prod_{i=1}^{n}P^{(k)}_{j_i}(w_i)\Big)\cdot P^{(k)}_{r_{j}}(\#)\big|\big|_2^2
\] 
Then, the following chain of equations and inequalities holds true.\\
\[\begin{array}{c c c c}
p_A(w) &=& \sum_{r_j \in F}\sum_{j_1, \ldots, j_n} \big|\big|
\pi_0 \cdot\Big(\prod_{i=1}^{n} P^{(k)}_{j_i}(w_i)\Big)\cdot P^{(k)}_{r_j}(\#)\big|\big|_2^2 & \\
&=& \sum_{j_1, \ldots, j_n\in \{\nearrow, \searrow\}} \big|\big|
\pi_0\cdot\Big(\prod_{i=1}^{n} P^{(k)}_{j_i}(w_i)\Big)\cdot P^{(k)}_{\text{acc}}(\#)\big|\big|_2^2 & \\
&\geq & \big|\big| \pi_0\cdot \Big(\prod_{i=1}^{n}P^{(k)}_{\nearrow}(w_i)\Big)\cdot P^{(k)}_{\text{acc}}(\#)\big|\big|_2^2 & ||\cdot||_2^2\geq 0 
 \\
& = & \Big(\pi_0 \cdot \prod_{i=1}^{n}P^{(k)}_{\nearrow}(w_i)\Big)_{k+1}^2
& \text{ by definition of } P^{(k)}_{\text{acc}}(\#)
  \\
& \geq & 2^{-2k} & \text{ Lemma \ref{lemma:ampiezzapositivalimitata}}\\
\end{array}
\]
\emph{Case 2: $w\not\in L[a_1, \ldots, a_k]$}. 
From Lemma $\ref{lemma:nonperturbazione}$ we have the following implication: 
\[
\text{if } w \not\in L[a_1, \ldots, a_k] \text{ then } \sum_{j_1, \ldots, j_n\in \{\nearrow,\searrow\}} \big|\big|
 \Big(\pi_0\cdot \prod_{i=1}^{n} P^{(k)}_{j_i}(w_i) \Big)_{k+1} \big|\big|_2^2 = 0\] hence  
\[\text{ if } w \not\in L[a_1, \ldots, a_k] \text{ then } p_A(w)=0 \]
This imply the thesis.
\end{proof}
Since the class \lipt\ is the boolean closure of languages of the form $L[a_1, \ldots, a_k]$, and $\lmo(\Sigma)$ is a boolean algebra, Theorem \ref{th:PTtoMON} implies that all literally idempotent piecewise testable languages can be recognized by \mons.
This concludes the proof of our main result, as stated in Theorem \ref{th:MAIN}.
\section{Algorithmic and logical conclusions}
Theorem \ref{th:MAIN} allows us to prove the existence of a polynomial time algorithm for deciding $\lmo(\Sigma)$ membership. In this way, we solve the polynomial-time 
 decidability question for $\lmo(\Sigma)$.
\begin{theorem}\label{th:LMOalgo}
Given a regular language $L\in\Sigma^*$, the problem of determining whether $L\in\lmo(\Sigma)$ is decidable in time $O((|Q|+|\Sigma|)^2)$, where $|Q|$ is the size of the minimal deterministic automaton for $L$.
\end{theorem}
\begin{proof}
This algorithm first constructs the minimal deterministic automaton $A_L$ for $L$ in time $O(|Q|\log(|Q|))$ as shown in \cite{H71}. Then, in time $O(|Q| + |\Sigma|)$, it checks whether $L$ is literally idempotent by visiting all the vertices and edges in the graph of $A_L$. Finally, it verifies whether $L$ is piecewise testable in time $O((|Q|+|\Sigma|)^2)$ with the technique shown in \cite{T01}. The fact that $\lmo(\Sigma)=\lipt(\Sigma)$ completes the proof.
\end{proof}
We conclude this last section with some remarks from mathematical logic. The piecewise testable languages are known in literature to be exactly those languages definable in the boolean closure of the existential first-order logic $\Sigma_1[<]$, denoted by $\text{Bool}(\Sigma_1[<])$ \cite{Si75}. 
We would like to characterize it here the logical definability of \lmo\, languages as well. With this in mind, we set a first-order 
syntax with atomic formulae of the following type:
\[
\lambda(x)=a \text{ and }  x < y \text{ and } \top
\] 
where $x,y$ are variables, $a\in\Sigma$ is a letter and $\top$ is a constant which means \emph{true}. If $\varphi, \psi$ are 
first-order formulae, then 
\[
\neg\varphi \text{ and } \varphi\vee \psi \text{ and } \exists x \varphi
\] are first-order formulae as well. In order to introduce a semantics in terms of words of $\Sigma^*$, we let the variables range 
over positions of words. A variable not quantified is called a \emph{free} variable, then a \emph{sentence} is a formula without 
free variables. Let the free variables of the formula $\varphi$ range in a subset of $\{x_1, \ldots, x_n\}$. Each $x_i$ is associated 
with a position $j_i$ of $w$. In this way, we can consider the truth value of $\varphi$ and denote it by 
$w, j_1, \ldots, j_n \models\varphi$. To any first-order sentence $\varphi$ we associate the language 
$L(\varphi) = \{w\in\Sigma^* | w\models \varphi \}$. We recall that the existential first order fragment $\Sigma_1[<]$ is given by 
first-order formulae in which we allow just one block of existential quantifiers and no blocks on universal quantifiers, in the  
prenex-normal form. Now let us consider the language $L[a_1, \ldots, a_k]$ given by \[\Sigma^* a_1 \Sigma^* a_2 \Sigma^*\cdots \Sigma^* a_k \Sigma^*, \text{ for } a_1, a_2, \ldots, a_k\in\Sigma \text{ and } 
 a_i \neq a_{i+1} \text{ for every } 1 \leq i < k \]
 then we shall define the following $\Sigma_1[<]$ formulae, we call them the \emph{easy} 
 formulae of the existential first-order logic: 
\[
\varphi[a_1, \ldots, a_k] := \exists x_1\cdots \exists x_k \left(\bigwedge_{i=1}^{k-1} x_i < x_{i+1} \wedge \bigwedge_{i=1}^{k} \lambda(x_i) = a_i \wedge 
\bigwedge_{i=1}^{k-1} \neg(\lambda(x_i) = a_{i+1})\right)
\]
Then clearly $\varphi[a_1, \ldots, a_k]\in\Sigma_1[<]$ and, by induction on $k$, it is possible to prove that $L(\varphi[a_1, \ldots, a_k])=L[a_1, \ldots, a_k]$. This proves the following logical characterization of \lmo\, :
\begin{proposition} $\lmo$ is exactly the class of languages definable by boolean combinations of \emph{easy} formulae 
of the existential first-order logic $\Sigma_1[<]$.\\
\end{proposition}
Also, we would like to take into account the \emph{linear temporal logic without the next operator} (LTLWN). 
The syntax and semantics of LTLWN is the same as the already presented semantics for first-order formulae, 
but we also consider the
$\emph{until}$ binary operator $\mathbb{U}$ whose semantics is defined as follows.
\[
w \models \varphi_1 \mathbb{U} \varphi_2 \iff \exists i\in\mathbb{N}\text{ s.t. } w_i \models \varphi_2\text{ and }
\forall 1 \leq j < i : w_j \models \varphi_1
\] Let $\Sigma$ be an alphabet and let $\Gamma_1, \ldots, \Gamma_k\subseteq \Sigma$ be non-empty sets. 
Define $\varphi_{\Gamma_i}:=\bigvee_{b\in\Gamma_i} b$ and $\varphi_{\epsilon}=\neg\varphi_{\Sigma}$. 
As Kl\'ima and Pol\'ak in \cite{KP08}, we shall consider the following formulae of LTLWN $\varphi[B_1, \ldots, B_k]$ and call them the \emph{easy} formulae of LTLWN:
\[
\varphi([\Gamma_1, \ldots, \Gamma_k]) := \varphi_{\Gamma_1}\mathbb{U}(\varphi_{\Gamma_2}\mathbb{U}(\cdots (\varphi_{\Gamma_n}\mathbb{U}\varphi_{\epsilon}))\cdots )
\]
It is possible to prove that $L(\varphi_{\Gamma_i}) = \Gamma_i\Sigma^*$ and that $\varphi_{\epsilon} = \epsilon$. 
Moreover that $L([\Gamma_1, \ldots, \Gamma_k]) = \Gamma^*_1\cdots \Gamma^*_k$. This imply that $L$ is definable 
as a boolean combination of $\emph{easy}$ formulae of LTLWN if and only if it is a 
literally idempotent piecewise testable language, as proved in \cite{KP08}. 
We get the following immediate corollary, which concludes our work on \mons\,.
\begin{proposition}
$\lmo$ is exactly the class of languages definable by boolean combinations of $\emph{easy}$ formulae of the linear temporal 
logic without the next operator LTLWN.
\end{proposition}


\appendix
\section{Appendix A}
\subsection{Proof of Proposition \ref{prop:linearRep} (i)}
The formal power series generated by a \mon\, $A$ on $\Sigma$ with $m$ states admits a linear representation 
$\langle \xi, (P(c))_{c\in\Sigma}, \eta \rangle$ where $||\xi|| = 1$, $P(c)$ is a projector for all $c\in\Sigma$ and $||\eta||\leq\sqrt{m}$.
\begin{proof}
The proof is almost identical to the one given over $(\Sigma, E)$ in \cite{BMP10}.
Given a matrix $M\in\mathbb{C}^{m\times m}$ with rows $r_1, \ldots, r_m$, let $\varphi(M)$ 
be the vector $(r_1, \ldots, r_m)\in\mathbb{C}^{1\times m^2}$. 
Let $A$ be a \mon\, over $\Sigma$ with with $\pi_0$ as initial state, the observable $O_c$ described by projectors 
$P_1(c), \ldots, P_k(c)$ for any $c\in\Sigma$, the observable $O_{\sharp}$ described by projectors 
$P_1(\sharp), \ldots, P_s(\sharp)$ and $F\subseteq\{r_1, \ldots, r_s\}$. 
For any $w:=w_1\cdots w_n\in\Sigma^*$, Proposition \ref{prop:sumofsquares} imply 
\[p_A(w) = \text{Tr}\left(\sum_{r_j\in F}\sum_{k_1, \ldots, k_n} 
||\pi_0 P_{k_1}(w_1)\cdots P_{k_n}(w_n)P_j(\sharp) ||^2\right)\]
Consider the formal series $\psi:\Sigma^*\rightarrow \mathbb{C}$ whose linear representation is given by: 
\[ \left(\pi_0 \otimes \pi_0^*, \Big(P(c):=\sum_j P_j(c)\otimes P_j(c)^*\Big)_{c\in\Sigma}, 
\eta := \varphi(I)\cdot \sum_{r_k\in F} P_k(\sharp) \otimes P_k(\sharp)\right)\] where $I$ is the identity matrix. Then 
\[
\begin{array}{c c c}
\psi(w_1\cdots w_n) &=&  (\pi_0\otimes \pi_0^*)(P(w_1)\ldots P(w_n))(\varphi(I)\cdot \sum_{r_k\in F}P_k(\sharp)\otimes P_k(\sharp)^*)^T \\ 
 &=& \sum_{r_k\in F, j_1, \ldots, j_n} \Big[ \pi_0 P_{j_1}(w_1)\cdots P_{j_n}(w_n)P_{r_k}(\sharp)\otimes  \\ 
 & & \pi_0^* P_{j_1}(w_1)^*\cdots P_{j_n}(w_n)^*P_{r_k}(\sharp)^* \Big]\varphi(I)^T\\
 &=& \sum_{r_k\in F, j_1, \ldots, j_n}\sum_{j}  (\pi_0 P_{j_1}(w_1)\cdots P_{j_n}(w_n)P_k(\sharp))_j\cdot \\
 & & (\pi_0 P_{j_1}(w_1) \cdots P_{j_n}(w_n)P_k(\sharp))^*_j\\
 &=&  \sum_{r_k\in F, j_1, \ldots, j_n} || \pi_0 P_{j_1}(w_1)\cdots P_{j_n}(w_n)P_k(\sharp)||^2 \\
 &=& p_A(w_1\cdots w_n) \\
\end{array}
\]
Observe that $P(c):=\sum_j P_j(c)\otimes P_j(c)^*$ is Hermitian and idempotent, hence it is a projector. Observe further that $|| \xi ||=||\pi_0\cdot \pi_0^*|| = || \pi_0 ||\cdot || \pi_0 ||=1$ and 
$||\eta|| = ||\varphi(I)\cdot \sum_{r_k\in F} P_k(\sharp)\otimes P_k(\sharp)^* || \leq ||\varphi(I)|| = \sqrt{m}$.
\end{proof}
\subsection{Proof of Proposition \ref{prop:linearRep} (ii)}
The class $\pmo(\Sigma)$ is closed under the operations of Hadamard product and $f$-complement.
\begin{proof}
The proof is almost identical to the one given over $(\Sigma, E)$ in \cite{BMP10}.
Let $A,A'$ be two \mons\, $A=\langle \Sigma\cup\{\#\}, (O_c)_{c\in\Sigma\cup\{\#\}, \pi, F} \rangle$ and 
$A'=\langle \Sigma\cup\{\#\}, (O'_c)_{c\in\Sigma\cup\{\#\}, \pi', F'} \rangle$. Without loss of generality, we assume that 
$V(O_c\otimes O'_c) = V(O_c)\otimes V(O'_c)$ for all $c\in\Sigma\cup\{\#\}$. Let us consider the automaton
\[
A'':=\langle \Sigma\cup\{\#\}, (O_c\otimes O'_c)_{c\in\Sigma\cup\{\#\}}, \pi\otimes \pi', F\times F' \rangle
\]
For any $w\in\Sigma^*$, by induction on the length $n:=|w|$ and the basic property that 
$(A\otimes B)(C\otimes D) = AC\otimes BD$, we see that $\sigma''(w)=\sigma(w)\otimes \sigma'(w)$, for all $w\in\Sigma^*$. 
Hence the following holds true for every $w\in\Sigma^*$:
\[ \displaystyle
\begin{array}{c c l}
p_{A''}(w) &=& \text{Tr}\Big(\sum_{(r_j,r'_k)\in F\times F'}^{} (P_{r_j}(\#)\otimes P'_{r'_k}(\#))\sigma''(w) (P_{r_j}(\#)\otimes P'_{r'_k}(\#))\Big)\\
&=& \text{Tr}\Big(\sum_{r_j\in F}\sum_{r'_k\in F'} P_{r_j}(\#)\sigma(w) P_{r_j}(\#) \otimes P'_{r'_k}(\#)\sigma'(w) P'_{r'_k}(\#) \Big)\\
&=& \text{Tr}\Big(\Big[ \sum_{r_j\in F} P_j(\sharp) \sigma(w) P_j(\sharp)\Big]\otimes \Big[\sum_{r'_k\in F} P_{r'_k}(\sharp) \sigma'(w) P_{r'_k}(\sharp) \Big] \Big)\\
&=& p_A(w)\cdot p_{A'}(w) \\
\end{array}
\]
and this proves closure under Hadamard product. To prove closure under complement, let $B$ the \mon\, defined to be equals 
to the \mon\, $A$ but with $V(O_{\#}^B):= V(O_{\#}^A)\setminus F_A$. Then $p_B(w) = 1-p_A(w)$ for every $w\in\Sigma^*$.

\end{proof}

\subsection{Proof of Proposition \ref{prop:linearRep} (iii)}
Let $L\in\lmo(\Sigma)$ and $\chi_L$ be its characteristic function. For any $\epsilon > 0$, there exists $\phi\in\pmo(\Sigma)$ such that $|\phi(w)-\chi_L(w)|<\epsilon$, for all $w\in\Sigma^*$.
\begin{proof} 
The proof is almost identical to the one given over $(\Sigma, E)$ in \cite{BMP10}.
Let $L\in\lmo(\Sigma)$ and let $A=\langle \Sigma\cup\{\sharp\}, (O_c)_{c\in\Sigma\cup \{\sharp\}}, \pi_0, F \rangle$ be a \mon\, over $\Sigma$ inducing a probabilistic event $p_A$ such that $L=\{w\in\Sigma^* | p_A(w)>\lambda\}$ for some 
$\lambda\in [0, 1)$ and there exists $\delta > 0$ such that, for all $w\in\Sigma^*$ it holds $|p_A(w)-\lambda|\geq \delta$.
Then fix an integer $N\in\mathbb{N}_0$. Without loss of generality, for every $c\in\Sigma\cup\{\sharp\}$, we 
can define new observables $O^{(1)}_c, \ldots, O^{(N)}_c$ such that $O^{(i)}_c$ possesses the same set of projectors as $O_c$ but different eigenvalues, so that $V(O^{(1)}_c\otimes\cdots\otimes O^{(N)}_c)=V(O^{(1)}_c)\times \cdots \times V(O^{(N)}_c)$. 
Moreover, for $1\leq i\leq N$, let $F^{(i)}$ be the set of results of $O^{(i)}_{\sharp}$ corresponding to results of $O_{\sharp}$ in F. 
Then, the \mon\, $A^{(i)} = \langle \Sigma\cup\{\sharp\}, (O^{(i)}_c)_{c\in\Sigma\cup\{\sharp\}}, \pi_0, F^{(i)} \rangle$ over $\Sigma$ satisfies $p_{A^{(i)}}=p_A$. 
Now, consider the \mon\, over $\Sigma$ defined as follows 
\[
A_N = \langle \Sigma\cup\{\sharp\}, (O^{(1)}_c\otimes \cdots \otimes O^{(N)}_c)_{c\in\Sigma\cup\{\sharp\}}, 
\pi_0\otimes\cdots\otimes\pi_0, F_N\rangle
\] where $(r_{j_1}^{(1)}, \ldots, r_{j_n}^{(N)})$ belongs to $F_N$ if and only if there exist at least $\lambda N$ indexes $i \in \{1, \ldots, N\}$ such 
that $r^{(i)}_{j_i}\in F^{(i)}$. Then it holds that \[p_{A_N}(w) = \sum_{k\geq \lambda N} \binom{N}{k} p_A^k(w)(1-p_A(w))^{N-k} = \text{Pr}\Big[ \frac{\sum_{i=1}^N X_i}{N} \geq \lambda \Big]\] for every $w\in\Sigma^*$, where $X_i$'s are i.i.d. random variables over $\{0, 1\}$ with $\text{Pr}(X_i=1)=p_A(w)$. If $w\not\in L$, then $p_A(w)\leq \lambda -\delta$. Thus by H\"offdings' inequality, it holds that 
\[\text{Pr}\Big[ \frac{\sum_{i=1}^N X_i}{N} \Big]\leq \text{Pr}\Big[ \frac{\sum_{i=1}^N X_i}{N}-p_A(t) \geq \delta \Big]\leq e^{-2\delta^2 N}\] so we get $p_{A_N}(w)\leq e^{-2\delta^2 N}$. If $w\in L$, by the same reasoning we obtain $p_{A_N}(w)\geq 1-e^{-2\delta^2 N}$. This imply 
that for any $\epsilon > 0$, for every $N$ such that $\epsilon \geq e^{-2\delta^2 N}$, we have $|\chi_L(w)-p_{A_N}| \leq \epsilon$, where $\chi_L$ is the characteristic function of $L$.
\end{proof}
\subsection{Proof of Proposition \ref{prop:linearRep} (iv)}
Given a formal series $\phi:\Sigma^*\rightarrow [0,1]$, let $\langle \xi, (P(c))_{c\in\Sigma}, \eta \rangle$ be a linear representation of $\phi$, where $||\xi|| = 1$ and $P(c)$ is a projector for $c\in\Sigma$, $P(\epsilon)=I$. Suppose that $|\phi(w)-\lambda|\geq \delta > 0$ for all $w\in\Sigma^*$, and let $L:=\{w\in\Sigma^* | \phi(w)>\lambda\}$. 
\begin{enumerate}
\renewcommand\labelenumi{\theenumi}
\renewcommand{\theenumi}{(\roman{enumi})}
\item $T$ is a regular language on $\Sigma^*$.
\item there exists a finite-state automaton $\langle \Sigma, Q, (\underline{\delta}_c)_{c\in\Sigma}, q_0, F \rangle$ recognizing $L$ such that, for any $w,u\in\Sigma^*$, the following holds: 
$\underline{\delta}_w(q_0)\neq \underline{\delta}_u(q_0)$ implies $||\xi P(w) - \xi P(u)||\geq \frac{\delta}{|| \eta ||}$.
\end{enumerate}
\begin{proof} The proof technique is classical and it is almost identical to the one given in \cite{BMP10} for formal series over 
$\text{FI}(\Sigma, E)$. Consider the automaton $B$ whose (possibly infinite) state set is $\overline{Q}:=\{\xi P(w) | w \in\Sigma^*\}$, the 
transition on $c\in\Sigma$ is $\overline{\delta}_c(\xi P(w)):=\xi P(wc)$, the initial state is $\xi$ and the set of final states is 
$\overline{F} := \{\xi P(w) | \xi P(w)\eta^T > \lambda\}$. Then $w\in L$ if and only if $\overline{\delta}(\xi, w)\in\overline{F}$. Now define the binary 
relation $\sim \subseteq\overline{Q}\times\overline{Q}$ as $\xi P(w) \sim \xi P(u)$ if and only if there exists $x_1, \ldots, x_M\in\Sigma^*$ such 
that $x_1=w, x_M=u$ and $||\xi P(x_{i+1}) - \xi P(x_{i})||\leq \frac{\delta}{||\eta||}$ for every $1\leq i < M$. 
Since $P(c)$'s are a projectors, they can only decrease distances, then $\sim$ is a congruence. Consider automaton $B/\sim$, 
this automaton recognizes $L$. Also, for any $w,u\in\Sigma^*$, $[\xi P(w)]_{\sim}\neq [\xi P(u)]_{\sim}$ imply 
$||\xi P(w) - \xi P(u)|| \geq \frac{\delta}{||\eta||}$ and this proves the inequality stated in (ii). The cardinality of $\Sigma^*$ is denumerable, for $i\in\mathbb{N}$ let $w_i$ be the representative of the $i$-th equivalence class $[\xi P(w_i)]$. Observe that the $\xi P(w_i)$'s lie inside the sphere of radius $1$, which is a compact set, but if $i\neq j$ then $||\xi P(w_i) - \xi P(w_j) || \geq \frac{\delta}{||\eta||}$. This means that there are a finite number of equivalence classes of $\sim$, since otherwise they breach the unitary sphere. The thesis follows.
\end{proof}
\subsection{Proof of Proposition \ref{prop:linearRep} (v)}
Let $\Sigma$ be a finite non-empty set. The class $\lmo(\Sigma)$ is closed under intersection, complement and union.
\begin{proof} The proof is almost identical to the one given over $(\Sigma, E)$ in \cite{BMP10}.
Let $L_1, L_2 \in \lmo(\Sigma)$. By Proposition \ref{prop:linearRep} (iii) there are $\phi_1, \phi_2\in\pmo(\Sigma)$ such that 
$|\phi_1(w)-\chi_{L_1}(t)|<\frac{1}{4}$ and $|\phi_2(w)-\chi_{L_2}(w)|<\frac{1}{4}$, for all $w\in\Sigma^*$. By Proposition \ref{prop:linearRep} (ii), 
$\phi_1\odot\phi_2\in\pmo(\Sigma)$. The probabilistic event $\phi_1\odot\phi_2$, with isolated cut point $\frac{1}{2}$, defines the language 
$L_1\cap L_2$. Moreover, $\mathbf{1}-\phi_1\in\pmo(\Sigma)$ and the complement $L^c_1$ is defined by $\mathbf{1}-\phi_1$ with isolated 
cut point $\frac{1}{2}$. Therefore $\lmo(\Sigma)$ is a boolean algebra.
\end{proof}
\subsection{Proof of Proposition \ref{prop:linearRep} (vi)}
$\lmo(\Sigma)$ is a boolean algebra of regular languages in $\Sigma^*$ with finite variation. In particular, if $L$ is a 
language recognized by a \mon\, over $\Sigma$ with $m$ states and isolation $\delta$, then 
$\sup_{x\in\Sigma^*}\text{var}_L(x)\leq \frac{m}{\delta^2}$.
\begin{proof} The proof is almost identical to the one given over $(\Sigma, E)$ in \cite{BMP10}. 
Combining Proposition \ref{prop:linearRep} (iv) with Proposition \ref{prop:linearRep} (v), $\lmo(\Sigma)$ is a 
is a boolean algebra of regular languages. Now, let $L\in\lmo(\Sigma)$, there exists a cut-point $\lambda$, a real $\delta>0$ and 
a \mon\, over $\Sigma$, let denote it by $A=\langle \Sigma\cup\{\sharp\}, (O_c)_{c\in\Sigma\cup\{\sharp\}}, \pi_0, F \rangle$, with $m$ pure states 
such that $L=\{w\in\Sigma^* | p_A(w) > \lambda \}$ and $|p_A(w)-\lambda | \geq \delta$ for all $w\in\Sigma^*$. By Proposition 
\ref{prop:linearRep} (i) there exists a linear representation $\langle \xi, (P(c))_{c\in\Sigma}, \eta \rangle$ of $p_A$, where 
$|| \xi || =1, P(c)\in\mathbb{C}^{m^2\times m^2}$ is a projector for every $c\in\Sigma$ and $||\eta||\leq \sqrt{m}$. By Proposition \ref{prop:linearRep} (iv), there 
exists an automaton recognizing $L$ satisfying the following property: for any $w,u\in\Sigma^*$, if 
$\underline{\delta}_w(q_0)\neq \underline{\delta}_u(q_0)$ then $||\xi P(w) - \xi P(u) || \geq \frac{\delta}{||\eta||}$. As a consequence, by 
considering the minimum automaton $\langle \Sigma, Q, \overline{\delta}, q_0, F \rangle$ for $L$, we have:
\[
\overline{\delta}(q_0, w)\neq \overline{\delta}(q_0, u) \text{ implies } ||\xi P(w) - \xi P(u) || \geq \frac{\delta}{||\eta||}
\]
Since $P(c)$ is a projector, we have $||\xi P(w)||^2 = ||\xi P(wc)||^2 + ||\xi P(w) - \xi P(wc)||^2$, for all $w\in\Sigma^*$.
Therefore, if $\delta(q_0, w_1\cdots w_{k-1})\neq \delta(q_0, w_1\cdots w_{k})$ for $w_i$'s in $\Sigma$, then it follows that:
\[
\begin{array}{c c c}
||\xi P(w_1\cdots w_k) ||^2 &=& || \xi P(w_1\cdots w_{k-1}) ||^2 - || \xi P(w_1\cdots w_{k-1}) - \xi P(w_1\cdots w_k) ||^2	\\
& \leq & ||\xi P(w_1\cdots w_{k-1}) ||^2 - \frac{\delta^2}{||\eta^2||}\\
& \leq & ||\xi P(w_1\cdots w_{k-1}) ||^2 - \frac{\delta^2}{m}\\
\end{array}
\]
By iterating this way, the following holds
\[
0\leq ||\xi P(w_1\cdots w_k) ||^2 \leq ||\xi ||^2 - \text{var}_L(w_1\cdots w_k)\frac{\delta^2}{m} = 1-\text{var}_L(w_1\cdots w_k)\frac{\delta^2}{m}, 
\]
This imply $\text{var}_L(w_1\cdots w_k)\leq \frac{m}{\delta^2}$, for all $w_1\cdots w_k\in\Sigma^*$.
\end{proof}

\end{document}